\documentclass[11pt]{article}
\usepackage{amsmath,amssymb,amsthm}
\usepackage[color,all]{xy}

\title{
Martingale Cohomology, Holonomy, and Homological Arbitrage\\
\large From Martingale Cohomology to Loop Effects in Financial Markets
}

\author{Takanori Adachi\footnote{Graduate School of Management, Tokyo Metropolitan University, Marunouchi Eiraku Bldg. 18F, 1-4-4 Marunouchi, Chiyoda-ku, Tokyo 100-0005 Japan\\Email: taka.adachi@tmu.ac.jp}
\thanks{This work was supported by JSPS KAKENHI Grant Number 24K04941.}
}

\date{\today}

\newtheorem{definition}{Definition}
\newtheorem{remark}{Remark}
\newtheorem{proposition}{Proposition}
\newtheorem{lemma}{Lemma}

\newcommand{\Prob}{\mathbf{Prob}}
\newcommand{\mpProb}{\mathbf{mpProb}}
\newcommand{\Ban}{\mathbf{Ban}}
\newcommand{\Hol}{\mathrm{Hol}}
\newcommand{\Exp}{\mathcal{E}}

\newcommand{\bdelta}{\mathbf{\delta}}

\newcommand{\prob}{\mu}
\newcommand{\Id}{\mathrm{id}}

\begin{document}
\maketitle

\begin{abstract}
\noindent
We introduce a transport cohomological framework for categorical
filtrations. Given a contravariant filtration
$
F:\mathcal T^{op}\to\Prob
$
on a small category \(\mathcal T\), conditional expectation induces
transport operators between local probabilistic states. Using the
simplicial structure of the nerve \(N_\bullet(\mathcal T)\), we
construct simplex-local cochain complexes associated with parametrized
simplices and study their transport cohomology.

The resulting framework naturally produces loop effects and holonomy
structures. In particular, transport around closed simplicial histories
may generate nontrivial probabilistic distortions, even when the initial
and terminal objects coincide. The associated holonomy operators encode
global transport effects between probabilistic states and detect
obstructions generated by loop transport.

This leads to the notion of homological arbitrage, understood as a
global transport phenomenon emerging from probabilistic distortion along
loops. From this viewpoint, the essential source of loop effects is the
probabilistic distortion generated by transport around closed simplicial
histories.

The present framework is structurally analogous to parallel transport
and holonomy in differential geometry, providing a geometric viewpoint
on categorical filtrations and probabilistic transport structures.
\end{abstract}

\section{Introduction}
\label{sec:introduction}

The present paper explores a categorical and cohomological viewpoint on
filtrations in probability theory. While filtrations are fundamental
objects in stochastic analysis and mathematical finance, their study
from the viewpoint of category theory and transport geometry remains
largely undeveloped.

The purpose of this paper is to introduce a transport cohomological
framework for categorical filtrations and to investigate global loop
effects arising from probabilistic transport structures.

Let
$
F:\mathcal T^{op}\to\Prob
$
be a filtration defined as
a contravariant functor on a small category
\(
\mathcal T
\),
where
\(
\Prob
\)
denotes the category of probability spaces with null-preserving
measurable maps. 
This formulation extends classical filtrations while allowing for a richer class of information flows,
including branching and recombination.
To the best of our knowledge, with a few exceptions \cite{adachi_2025a},
such a clear formulation has rarely been systematically established in the field of mathematical finance.
Associated with such a filtration, conditional expectation 
induces transport operators between local probabilistic states. 
The resulting transport structures naturally interact with the simplicial geometry of the nerve
\(
N_\bullet(\mathcal T)
\).

A first observation is that martingale structures admit a natural
categorical interpretation. In particular, the martingale condition may
be expressed through the kernel of a transport differential operator.
However, higher-order transport phenomena exhibit a fundamentally local
character and are not naturally captured by global cochain systems over
the entire category
\(
\mathcal T
\).

To describe these higher transport structures, we introduce
simplex-local cochain complexes attached to parametrized simplices
$
\sigma\in N_k(\mathcal T).
$
The resulting \(\sigma\)-gauge complexes encode transport data along
simplicial histories and satisfy natural cochain conditions induced by
conditional expectation transport.

This local viewpoint leads naturally to loop effects and holonomy
phenomena. Even when a simplicial loop begins and ends at the same
object of
\(
\mathcal T
\),
transport around the loop may generate a nontrivial probabilistic
distortion. The associated holonomy operators describe accumulated
transport effects between probabilistic states and detect global
obstructions encoded in the transport history of the simplex.

From this perspective, the essential source of loop effects is the
probabilistic distortion generated by transport around closed simplicial
histories. This motivates the notion of homological arbitrage developed
in the present paper. Here, homological arbitrage should be understood
not as a local pricing inconsistency, but as a global transport effect
emerging from loop transport and detected through cohomological and
holonomy structures.

The present framework is structurally close to holonomy theory and
parallel transport in differential geometry. In particular, the
transport operators introduced here behave analogously to probabilistic
parallel transports, while nontrivial holonomy represents a global
obstruction to probabilistic trivialization along loops.

Finally, we note that the holonomy perspective developed in this paper is closely related
to the loop-based arbitrage mechanisms studied in \cite{adachi_2026a},
where arbitrage arises from global phase-like effects.
The present work provides a cohomological framework that complements and clarifies these phenomena.

The paper is organized as follows. 
In Section~\ref{sec:filtrations},
we introduce categorical filtrations, distortion operators, and martingale structures. 
Section~\ref{sec:localTransportComplexes}
develops the local transport complexes associated with parametrized simplices and establishes the corresponding cochain conditions.
In Section~\ref{sec:homologicalArb},
we study transport cohomology, loops, holonomy, and homological arbitrage.
Section~\ref{sec:examples}
presents several examples illustrating transport consistency and nontrivial loop effects.

\section{Martingales on categorical filtrations}
\label{sec:filtrations}

In this section, we recall the notion of a categorical filtration and the corresponding martingale condition. We then explain why a naive higher-order extension in the original $\mu$-gauge fails to produce a cochain complex, which motivates the gauge-change construction developed in the next section.

\subsection{$\mathcal{T}$-filtrations}

Let $\mathcal{T}$ be a small category, regarded as a time domain.

\begin{definition}[\cite{ANR_2020b}]
A \emph{$\mathcal{T}$-filtration} is a contravariant functor
\begin{equation}
\label{eq:tFiltration}
F : \mathcal{T}^{op} \longrightarrow \Prob,
\end{equation}
where
$\Prob$
is the category whose objects are all probability spaces and arrows are measurable maps such that
their inverse preserve null-sets (\emph{null-preserving map}).
\end{definition}

Thus, for each object $t \in \mathcal{O}_{\mathcal T}$, we have a probability space
\begin{equation}
\label{eq:FtNotation}
F(t)=(\Omega_t,\mathcal F_t,\prob_t),
\end{equation}
and for each arrow $i:s\to t$ in $\mathcal T$, we have a measurable map
\[
F(i):F(t)\to F(s)
\]
satisfying
\[
\mu_t \circ F(i)^{-1} \ll \mu_s.
\]
Functoriality means
\[
F(1_t)=1_{F(t)},
\qquad
F(j\circ i)=F(i)\circ F(j)
\]
for composable arrows $s \xrightarrow{i} t \xrightarrow{j} u$.


\subsection{Conditional expectation and the density operator}
\label{sec:condExpAndDensityOp}

Let
\begin{equation}
\label{eq:condExpF}
\Exp : \Prob \longrightarrow \Ban
\end{equation}
denote the conditional expectation functor \cite{AR_2019},
characterized by the identity
\[
\int_B \mathcal{E}(\varphi)(f)\, d\mu_Y
=
\int_{\varphi^{-1}(B)} f\, d\mu_X, 
\qquad (\forall B\in\mathcal F_Y),
\]
for a $\Prob$-arrow 
$\varphi : X \to Y$
and
$f \in \Exp(X) :=  L^1(X)$. 

For each arrow $i:s\to t$ in $\mathcal T$, we write
\[
(\Exp\circ F)i : (\Exp\circ F)t \longrightarrow (\Exp\circ F)s
\]
for the corresponding conditional expectation operator.

\begin{definition}
An \emph{$F$-adapted process} is an element
\[
f=\{f_t\}_{t\in \mathcal O_{\mathcal T}}
\]
of
\begin{equation}
\label{eq:AdapF}
\mathbf{Adap}(F)
:=
C^0
:=
\prod_{t\in \mathcal O_{\mathcal T}} (\Exp\circ F)t.
\end{equation}
\end{definition}

\begin{definition}
An \emph{$F$-martingale} is an $F$-adapted process
\[
f=\{f_t\}_{t\in \mathcal O_{\mathcal T}}
	\in
\mathbf{Adap}(F)
\]
such that for every arrow $i:s\to t$ in $\mathcal T$,
\begin{equation}
\label{eq:fMartingale_paper}
(\Exp\circ F)i(f_t)=f_s\cdot dF(i),
\end{equation}
where
\begin{equation}
\label{eq:dF_def_paper}
dF(i)
	:=
(\Exp \circ F) i \big(1_{F(t)}\big)
	=
\frac{d(\prob_t\circ (F(i))^{-1})}{d\prob_s}
\in (\Exp\circ F)s .
\end{equation}
\end{definition}

Note that \(dF(i)=1\) whenever \(F(i)\) is measure-preserving.

\begin{lemma}
For composable arrows
\[
s\xrightarrow{i} t \xrightarrow{j} u
\]
in $\mathcal T$, one has
\begin{equation}
\label{eq:dF_comp_paper}
dF(j\circ i)= (\Exp\circ F)i\bigl(dF(j)\bigr).
\end{equation}
\end{lemma}

\begin{proof}
For every $A\in \mathcal F_s$,
\begin{align*}
\int_A dF(j\circ i)\, d\prob_s
&=
(\prob_u\circ (F(j\circ i))^{-1})(A) \\
&=
(\prob_u\circ (F(j))^{-1}\circ (F(i))^{-1})(A) \\
&=
\int_{(F(i))^{-1}(A)} d(\prob_u\circ (F(j))^{-1}) \\
&=
\int_{(F(i))^{-1}(A)} dF(j)\, d\prob_t .
\end{align*}
By the defining property of conditional expectation, this equals
\[
\int_A (\Exp\circ F)i(dF(j))\, d\prob_s,
\]
which proves the claim.
\end{proof}

\subsection{The first characterization of martingales}

Define
\[
C^1
:=
\prod_{i:s\to t \text{ in }\mathcal T} (\Exp\circ F)s
\]
and a map
\[
\bdelta^1 : C^0 \longrightarrow C^1
\]
by
\begin{equation}
\label{eq:delta1_mu_paper}
\bdelta^1(f)(i):=(\Exp\circ F)i(f_t)-f_s\cdot dF(i).
\end{equation}

Then the martingale condition is exactly the vanishing of \(\bdelta^1\).

\begin{proposition}
Let \(\mathbf{Mart}(F)\subset C^0\) denote the set of all \(F\)-martingales.
For \(f\in C^0\),
\[
f \text{ is an } F\text{-martingale}
\quad\Longleftrightarrow\quad
\bdelta^1(f)=0.
\]
Hence
\[
\mathbf{Mart}(F)=\ker \bdelta^1.
\]
\end{proposition}

\section{Local transport complexes}
\label{sec:localTransportComplexes}

We now explain how the first-order characterization above suggests a higher-order theory. A direct extension in the original $\mu$-gauge fails in general to define a cochain complex. This failure is repaired by a gauge change along simplices.

\subsection{The naive $\mu$-chain}

Let
$\Delta$
be the simplex category,
and
$
\delta^n_i : [n-1] \to [n]
$
and
$
\sigma^n_j : [n+1] \to [n]
$
are usual coface maps and codegeneracy maps in $\Delta$, respectively.

\begin{definition}
\label{defn:naiveMuChain}
Let
\[
\tau \in N_n(\mathcal T)
\]
be an $n$-simplex,
where 
$N(\mathcal T)$
is the nerve of $\mathcal T$,
and
$
N_n(\mathcal T) := N(\mathcal T)([n]).
$
Write
\[
\tau = 
\bigl(
\tau(0) \xrightarrow{\tau(*_{1})} \tau(1) \xrightarrow{\tau(*_{2})} \cdots \xrightarrow{\tau(*_{n})} \tau(n)
\bigr).
\]
\begin{enumerate}
  \item
\begin{equation}
\label{eq:Cmu_paper}
C^n
:=
\prod_{\tau \in N_n(\mathcal T)} (\Exp\circ F)(\tau(0)).
\end{equation}

\item
For \(c\in C^{n-1}\), 
\[
\bdelta^n : C^{n-1}\to C^n
\]
by
\begin{equation}
\label{eq:delta_mu_n_paper}
\bdelta^n(c)(\tau)
:=
\sum_{i=0}^n (-1)^i A_{i}^{\tau}\bigl(c(\tau\circ\delta_i^n)\bigr),
\end{equation}
where
\[
A_{i}^{\tau}
:
(\Exp\circ F)((\tau\circ\delta_i^n)(0))
\longrightarrow
(\Exp\circ F)(\tau(0))
\]
is given by
\begin{equation}
\label{eq:A_i_paper}
A_{i}^{\tau}(f)
:=
\begin{cases}
(\Exp\circ F)(\tau(*_{1}))(f) & \text{if } i=0,\\[4pt]
f\cdot dF(\tau(*_{1})) & \text{if } i=1,\\[4pt]
f & \text{if } i\ge 2.
\end{cases}
\end{equation}

\end{enumerate}
\end{definition}

This construction extends the operator \(\bdelta^1\), but in general it does not define a cochain complex.

Indeed, let
\[
\mathbf C:=\bdelta^{n+1}\circ \bdelta^n .
\]
Using the simplicial identity
\[
\delta_j^{n+1}\circ \delta_i^n = \delta_i^{n+1}\circ \delta_{j-1}^n
\qquad (0\le i<j\le n+1),
\]
one obtains
\[
\mathbf C(c)(\tau)
=
\sum_{0\le i<j\le n+1}
(-1)^{i+j}
D_{i,j}^{\tau}
\bigl(c(\tau\circ \delta_j^{n+1}\circ \delta_i^n)\bigr),
\]
where
\[
D_{i,j}^{\tau}
:=
\bigl(
A_{j}^{\tau}\circ A_{i}^{\tau\circ\delta_j^{n+1}}
\bigr)
-
\bigl(
A_{i}^{\tau}\circ A_{j-1}^{\tau\circ\delta_i^{n+1}}
\bigr).
\]

In particular, for \((i,j)=(0,1)\), one finds in general that
\[
D_{0,1}^{\tau}\neq 0.
\]
Hence
\[
\bdelta^{n+1}\circ \bdelta^n \neq 0
\]
in general, so the $\mu$-chain is not a cochain complex.

This failure motivates the introduction of a gauge-change method.

\subsection{$\mu$-gauge and $\sigma$-gauge}
\label{sec:muGaugeAndBetaGauge}

Let \(\sigma\in N_k(\mathcal T)\) be a fixed $k$-simplex:
\[
\sigma = (
	t_0
		\xrightarrow{i_1}
	t_1
		\xrightarrow{i_2}
	\cdots
		\xrightarrow{i_k}
	t_k
),
\]
and for \(\ell\in [k]\) write
\[
F_\ell
	:=
F(t_{\ell})
	=
(\Omega_{t_\ell},\mathcal F_{t_\ell},\prob_{t_\ell}),
	\qquad
\phi_\ell
	:=
F(i_{\ell})
	:
F_\ell \to F_{\ell-1}
	\quad
(1 \le \ell\le k).
\]

Thus the $k$-simplex \(\sigma\) determines a diagram in \(\Prob\),
\[
F_0 \xleftarrow{\phi_1} F_1 \xleftarrow{\phi_2}\cdots \xleftarrow{\phi_k} F_k,
\]
which we call the \emph{$\mu$-gauge}.

Now define recursively a modified family of measures
$\mu^{\sigma}_{t_{\ell}}$:
\begin{equation}
\label{eq:muSigma}
\mu^\sigma_{t_k}
	:=
\prob_{t_k},
\qquad
\mu^\sigma_{t_\ell}
	:=
\mu^\sigma_{t_{\ell+1}}
	\circ
F(i_{\ell+1})^{-1}
	\quad
(0 \le \ell < k).
\end{equation}
Then, we have
\[
\mu^{\sigma}_{t_{\ell}}
	\ll
\prob_{t_{\ell}}
	\quad
(0 \le \ell \le k).
\]

Using these measures, define a functor
\[
F^{\sigma} : \sigma^{op}\to \mpProb 
\]
by
\[
F^{\sigma}(t_{\ell})
=
F^{\sigma}_\ell
:=
(\Omega_{t_\ell},\mathcal F_{t_\ell},\mu^\sigma_{t_\ell}),
\qquad
F^{\sigma}(i_{\ell})
=
\phi^\sigma_\ell:=\phi_\ell,
\]
where
$\mpProb$
is the subcategory of $\Prob$
whose arrows are restricted to measure-preserving ones.

\medskip
\noindent
We call the resulting diagram
\[
F^{\sigma}_0
	\xleftarrow{\phi^\sigma_1}
F^{\sigma}_1
	\xleftarrow{\phi^\sigma_2}
\cdots
	\xleftarrow{\phi^\sigma_k}
F^{\sigma}_k
\]
the \emph{$\sigma$-gauge}.

\medskip
\noindent
Although the conditional expectation operator $\mathcal E$ is functorial
on $\Prob$
(see Section \ref{sec:condExpAndDensityOp}),
the presence of the density operator $dF$
introduces a multiplicative distortion along composable arrows. As a
result, the natural coboundary operators defined from $F$ do not, in
general, satisfy the cochain condition.

The construction of $F^\sigma$ provides a normalization along each
$k$-simplex $\sigma$, in which the induced maps become measure-preserving,
so that
\[
dF^\sigma(i_{\ell}) = 1.
\]
This removes the density distortion at the level of each simplex,
without altering the underlying probabilistic structure.

As we shall see, this normalization is precisely what allows the
resulting coboundary operators to satisfy
\[
\delta_{\sigma}^{n+1} \circ \delta_{\sigma}^n = 0,
\]
thereby giving rise to a well-defined cochain complex.

\subsection{The cochain complex with $\sigma$-gauge}

The normalization introduced in 
Section \ref{sec:muGaugeAndBetaGauge}
ensures that the density
terms vanish along each simplex, so that the coboundary operators
define a genuine cochain complex.

\begin{definition}
\label{defn:CnSigmaDeltaNsigma}
Let $n \ge 0$,
$
\theta \in \Delta([n], [k]).
$
and
$0 \le a \le b \le k$.
\begin{enumerate}
  \item
\begin{equation}
\label{eq:CnSigma}
C^n_{\sigma}
:=
\prod_{
	\theta \in \Delta([n], [k])
}
	L^1(F^{\sigma}_{\theta(0)}).
\end{equation}

\item
When $a < b$,
\begin{equation}
\label{eq:sigmaAB}
\langle \sigma \rangle_{a}^b
	:=
i_b \circ i_{b-1} \circ \cdots \circ i_{a+1} .
\end{equation}

\item
\begin{equation}
\label{eq:TsigmaAB}
T^\sigma_{a,b}
	:=
\begin{cases}
	(\mathcal E\circ F^\sigma)
	\bigl(\langle\sigma\rangle_a^b\bigr)
		:
	L^1(F^\sigma_b)\to L^1(F^\sigma_a)
			\quad
&\textrm{if } a < b,
		\\
\Id_{L^1(F^{\sigma}_a)}
			\quad
&\textrm{if } a = b.
\end{cases}
\end{equation}

\item
For 
$0 \le \ell \le n$,
\begin{equation}
\label{eq:AsigmaThetaEll}
A^{\sigma,\theta}_{\ell}
	:=
\begin{cases}
T^{\sigma}_{\theta(0), \theta(1)}
		& \text{if } \ell = 0,
			\\
\Id_{L^1\big(F^{\sigma}_{\theta(0)}\big)}
		& \text{if } \ell > 0.
\end{cases}
\end{equation}

\item
For \(c\in C_{\sigma}^{n-1}\),
\begin{equation}
\label{eq:delta_sigma_paper}
\delta_{\sigma}^n(c)(\theta)
	:=
\sum_{\ell=0}^n
	(-1)^\ell
	A^{\sigma,\theta}_{\ell}
	\bigl(c(\theta \circ \delta^n_{\ell})\bigr),
\end{equation}

\end{enumerate}
\end{definition}

Since \(\mathcal E\circ F^\sigma\) is functorial, we have
\begin{equation}
\label{eq:T_sigma_functoriality}
T^\sigma_{a,b}\circ T^\sigma_{b,c}
=
T^\sigma_{a,c}
\qquad
(0\le a\le b\le c\le k).
\end{equation}

Then, we have the following proposition.

\begin{proposition}
The operators $\delta^n_\sigma$ satisfy
\[
\delta^{n+1}_\sigma \circ \delta^n_\sigma = 0 .
\]
\end{proposition}

\begin{proof}
Let \(c\in C^{n-1}_\sigma\) and let
$
\theta\in \Delta([n+1],[k]).
$
By definition,
\[
(\delta^{n+1}_\sigma \circ \delta^n_\sigma)(c)(\theta)
=
\sum_{j=0}^{n+1}\sum_{i=0}^{n}
(-1)^{i+j}
\big(
A_j^{\sigma,\theta}
	\circ
A_i^{\sigma,\theta\circ\delta^{n+1}_j}
\big)
\bigl(
c(\theta\circ\delta^{n+1}_j\circ\delta^n_i)
\bigr).
\]

Using the simplicial identity
\[
\delta^{n+1}_j\circ\delta^n_i
=
\delta^{n+1}_i\circ\delta^n_{j-1}
\qquad (0\le i<j\le n+1),
\]
it is enough to prove that
\begin{equation}
\label{eq:A_sigma_pair_identity}
A_j^{\sigma,\theta}
	\circ
A_i^{\sigma,\theta\circ\delta^{n+1}_j}
=
A_i^{\sigma,\theta}
	\circ
A_{j-1}^{\sigma,\theta\circ\delta^{n+1}_i}
\end{equation}
for every \(0\le i<j\le n+1\).

We verify this by cases.

First suppose \(i>0\). Then \(j>0\) and \(j-1>0\). Hence all operators in
\eqref{eq:A_sigma_pair_identity} are identities, and the equality is immediate.

Next suppose \(i=0\) and \(j>1\). Since deleting the \(j\)-th vertex with
\(j>1\) does not affect the first two vertices, we have
\[
(\theta\circ\delta^{n+1}_j)(0)=\theta(0),
\qquad
(\theta\circ\delta^{n+1}_j)(1)=\theta(1).
\]
Therefore
\[
A_j^{\sigma,\theta}
	\circ
A_0^{\sigma,\theta\circ\delta^{n+1}_j}
=
\Id_{L^1\big(F^{\sigma}_{\theta(0)}\big)}
	\circ
T^\sigma_{\theta(0),\theta(1)}
=
T^\sigma_{\theta(0),\theta(1)}.
\]
On the other hand,
\[
A_0^{\sigma,\theta}
	\circ
A_{j-1}^{\sigma,\theta\circ\delta^{n+1}_0}
=
T^\sigma_{\theta(0),\theta(1)}
	\circ
\Id_{L^1\big(F^{\sigma}_{\big(\theta \circ \delta_0^{n+1}\big)(0)}\big)}
	=
T^\sigma_{\theta(0),\theta(1)}.
\]
Thus \eqref{eq:A_sigma_pair_identity} holds in this case.

Finally consider the case \((i,j)=(0,1)\). The left-hand side is
\[
A_1^{\sigma,\theta}
	\circ
A_0^{\sigma,\theta\circ\delta^{n+1}_1}
	=
\Id_{L_1\big(F^{\sigma}_{\theta(0)}\big)}
	\circ
T^\sigma_{\theta(0),\theta(2)}
	=
T^\sigma_{\theta(0),\theta(2)}.
\]
The right-hand side is
\[
A_0^{\sigma,\theta}
	\circ
A_0^{\sigma,\theta\circ\delta^{n+1}_0}
=
T^\sigma_{\theta(0),\theta(1)}
\circ
T^\sigma_{\theta(1),\theta(2)}.
\]
By \eqref{eq:T_sigma_functoriality}, this equals
\[
T^\sigma_{\theta(0),\theta(2)}.
\]
Hence \eqref{eq:A_sigma_pair_identity} also holds in the case
\((i,j)=(0,1)\).

Therefore every pair of terms in the double sum cancels with the
corresponding term of opposite sign. Consequently,
\[
\delta^{n+1}_\sigma\circ\delta^n_\sigma=0.
\]
\end{proof}

Hence we obtain a cochain complex
\[
0
	\to
C_{\sigma}^0
	\xrightarrow{\delta_{\sigma}^1}
C_{\sigma}^1
	\xrightarrow{\delta_{\sigma}^2}
C_{\sigma}^2
	\xrightarrow{\delta_{\sigma}^3}
\cdots.
\]

\section{Transport cohomology and loop effects}
\label{sec:homologicalArb}

In this section, we investigate the cohomological structures associated
with simplex-local transport complexes and study the global effects
generated by transport along loops.

Given a parametrized simplex
\(
\sigma\in N_k(\mathcal T)
\),
the associated \(\sigma\)-gauge complex encodes transport data along the
simplicial history represented by \(\sigma\). The corresponding
cohomology measures compatibility and obstruction phenomena arising from
probabilistic transport.

A central feature of the present framework is the appearance of loop
effects. Even when a loop begins and ends at the same object of
\(
\mathcal T
\),
transport around the loop may produce nontrivial probabilistic
distortion. The resulting holonomy operators describe accumulated
transport effects between probabilistic states and detect global
obstructions generated by closed transport histories.

From this viewpoint, homological arbitrage emerges as a global transport
phenomenon arising from probabilistic distortion along loops, rather
than from local inconsistencies at individual transitions.

\medskip
\noindent
Let
\[
\sigma
=
\left(
t_0
    \xrightarrow{i_1}
t_1
    \xrightarrow{i_2}
\cdots
    \xrightarrow{i_k}
t_k
\right)
\in N_k(\mathcal T)
\]
be a fixed simplex throughout this section.

\subsection{Transport cocycles and their consistency}

Recall that, in the \(\sigma\)-gauge, we constructed a cochain complex
\[
\xymatrix@C=24pt{
C^0_\sigma
    \ar[r]^-{\delta^1_\sigma}
&
C^1_\sigma
    \ar[r]^-{\delta^2_\sigma}
&
C^2_\sigma
    \ar[r]
&
\cdots
}
\]
satisfying
\[
\delta^{n+1}_\sigma
    \circ
\delta^n_\sigma
=
0 .
\]

\begin{definition}
For \(n\ge 0\), define
the space of
\emph{\(n\)-transport cocycles}
by
\[
Z^n_\sigma
:=
\ker \delta^{n+1}_\sigma ,
\]
and
the space of
\emph{\(n\)-transport coboundaries}
by
\[
B^n_\sigma
:=
\operatorname{im}
\delta^n_\sigma .
\]
The quotient
\[
H^n_\sigma
:=
Z^n_\sigma
/
B^n_\sigma
\]
is called the
\emph{transport cohomology}
of the simplex \(\sigma\).
\end{definition}

The cohomology groups
\(
H^\bullet_\sigma
\)
measure the extent to which transport structures along the simplex
\(\sigma\)
fail to be globally trivialized.

\vspace{0.5\baselineskip}

In particular,
elements of
\(
Z^1_\sigma
\)
describe transport systems which are compatible with simplicial
composition.

\begin{proposition}
\label{prop:transportCocycle}
Let
$
a
    \in
Z^1_\sigma
$
and
$
\theta
:
[2]
    \to
[k]
$
be an element of
\(
\Delta([2],[k])
\).
Then,
\begin{equation}
\label{eq:transportCocycle}
a(\theta\circ\delta^2_1)
=
T^\sigma_{r_0,r_1}
\bigl(
a(\theta\circ\delta^2_0)
\bigr)
+
a(\theta\circ\delta^2_2).
\end{equation}
\end{proposition}

\begin{proof}
Write
\[
\theta(0)
    =
r_0,
\qquad
\theta(1)
    =
r_1,
\qquad
\theta(2)
    =
r_2 .
\]

Then,
by definition,
\[
(\delta^2_\sigma a)(\theta)
=
A^{\sigma,\theta}_0
\bigl(
a(\theta\circ\delta^2_0)
\bigr)
-
A^{\sigma,\theta}_1
\bigl(
a(\theta\circ\delta^2_1)
\bigr)
+
A^{\sigma,\theta}_2
\bigl(
a(\theta\circ\delta^2_2)
\bigr).
\]

Since
\[
A^{\sigma,\theta}_0
=
T^\sigma_{r_0,r_1},
\qquad
A^{\sigma,\theta}_1
=
A^{\sigma,\theta}_2
=
\mathrm{id},
\]
we obtain
\[
(\delta^2_\sigma a)(\theta)
=
T^\sigma_{r_0,r_1}
\bigl(
a(\theta\circ\delta^2_0)
\bigr)
-
a(\theta\circ\delta^2_1)
+
a(\theta\circ\delta^2_2).
\]

On the other hand,
since
$
a
    \in
Z^1_\sigma ,
$
we have
\[
a(\theta\circ\delta^2_1)
=
T^\sigma_{r_0,r_1}
\bigl(
a(\theta\circ\delta^2_0)
\bigr)
+
a(\theta\circ\delta^2_2).
\]
\end{proof}

Proposition \ref{prop:transportCocycle}
means that the transport contribution associated with the composite
path from \(r_0\) to \(r_2\)
decomposes consistently into the transport from \(r_0\) to \(r_1\)
together with the transport from \(r_1\) to \(r_2\).

Thus,
\(1\)-transport cocycles represent transport systems compatible with
subdivision of paths inside the simplex \(\sigma\).


\medskip
\noindent
The cocycle condition introduced above admits a natural interpretation
as a compatibility condition for transports along subdivided paths.

Proposition \ref{prop:transportCocycle}
says that
the transport associated with the composite interval
\(
[r_0,r_2]
\)
decomposes into the transport from
\(
r_0
\)
to
\(
r_1
\)
and the transport from
\(
r_1
\)
to
\(
r_2
\).

Write
\begin{equation}
\label{eq:blacketK}
[k] = \big(
0 
	\xrightarrow{*_1}
1
	\xrightarrow{*_2}
\cdots
	\xrightarrow{*_k}
k
\big) .
\end{equation}
Then,
for 
$
\ell \in [k],
$
\[
t_{\ell} = \sigma(\ell),
	\qquad
i_{\ell} = \sigma(*_{\ell})
	\quad (\ell > 0) .
\]

\medskip
\noindent
Let
\begin{equation}
\label{eq:compStar}
*_{a,b}
	:=
*_b \circ *_{b-1} \circ \cdots \circ *_{a+1}
\quad
(0 \le a < b \le k) .
\end{equation}

\begin{proposition}
\label{prop:transportConsistency}
Let
$
a \in C_{\sigma}^1
$
and
$
\theta
	:
[m] \to [k]
$
be a strictly increasing function.
Then
\[
a\big(\theta(*_{0,m})\big)
=
\sum_{\ell=0}^{m-1}
T^\sigma_{
	\theta(0), \theta(\ell)
}
\Bigl(
a\big(
\theta(*_{\ell,\ell+1})
\big)
\Bigr).
\]

\end{proposition}
\begin{proof}
It is a straightforward result of the iterating use of
Proposition \ref{prop:transportCocycle}.
\end{proof}

Therefore,
a \(1\)-transport cocycle determines a transport system whose total
effect along a path depends only on the resulting composite transport
and is compatible with subdivision.

\vspace{0.5\baselineskip}

This behavior is analogous to parallel transport in differential
geometry:
local transport data combine coherently along concatenations of paths.

\begin{remark}
The transport cocycles considered here are local objects attached to a
fixed simplex \(\sigma\).
Consequently,
their consistency is formulated relative to the parametrized transport
structure of \(\sigma\),
rather than globally over the entire category \(\mathcal T\).
\end{remark}

\subsection{Loops and repeated vertices}

An important feature of the present construction is that the cochain
complex
\(
C^\bullet_\sigma
\)
is attached to a fixed parametrized simplex
\(
\sigma
\),
rather than merely to the image of the simplex inside
\(
\mathcal T
\).

Consequently,
different occurrences of the same object of
\(
\mathcal T
\)
along the simplex are treated as distinct positions.

\vspace{0.5\baselineskip}

More precisely,
let
\[
\sigma
=
\left(
t_0
    \xrightarrow{i_1}
t_1
    \xrightarrow{i_2}
\cdots
    \xrightarrow{i_k}
t_k
\right)
\in N_k(\mathcal T).
\]
Even if
\[
t_r=t_s
\qquad
(r\neq s),
\]
the vertices
\(r\)
and
\(s\)
are regarded as distinct occurrences inside the parametrized simplex
\(\sigma\).

Accordingly,
the coefficient spaces
\[
L^1(F^\sigma_r)
\qquad\text{and}\qquad
L^1(F^\sigma_s)
\]
are treated as distinct spaces,
even though they are associated with the same object of
\(
\mathcal T
\).

\vspace{0.5\baselineskip}

This distinction is particularly important in the case of loops.

\begin{definition}
A simplex
\[
\sigma
=
\left(
t_0
    \xrightarrow{i_1}
t_1
    \xrightarrow{i_2}
\cdots
    \xrightarrow{i_k}
t_k
\right)
\in N_k(\mathcal T)
\]
is called a
\emph{loop}
if
$
t_0=t_k .
$
\end{definition}


\vspace{0.5\baselineskip}

The distinction mentioned above reflects the path-dependent nature of transport.
Indeed,
the transport from the initial occurrence of
\(t_0\)
to its terminal occurrence along the loop
\(\sigma\)
may accumulate nontrivial effects,
even though the underlying object of
\(
\mathcal T
\)
returns to itself.

Thus,
the distinction between repeated occurrences of the same object is not
a defect of the formalism,
but rather an essential feature required to describe loop transport and
holonomy phenomena.

\vspace{0.5\baselineskip}

From this perspective,
the simplex
\(\sigma\)
should be understood not merely as a subset of objects of
\(
\mathcal T
\),
but as an ordered transport history.
The cochain complex
\(
C^\bullet_\sigma
\)
therefore records transport data along this history,
including possible nontrivial effects generated by loops.

\medskip
\noindent
The possibility of nontrivial transport effects along loops suggests a
cohomological mechanism for path-dependent inconsistencies.


\subsection{Holonomy}

We now consider transport around loops and the resulting return effects.

Let
\[
\sigma
=
\left(
t_0
    \xrightarrow{i_1}
t_1
    \xrightarrow{i_2}
\cdots
    \xrightarrow{i_k}
t_k
\right)
\in N_k(\mathcal T)
\]
be a loop, namely,
$
t_0=t_k .
$

The transport along the entire loop defines an operator
\[
T^\sigma_{0,k}
:
L^1(F^\sigma_k)
    \to
L^1(F^\sigma_0).
\]

Although
\(
t_0=t_k
\),
the corresponding probabilistic structures
\[
F^\sigma_0
=
(\Omega_{t_0},\mathcal F_{t_0},\mu^\sigma_0)
\]
and
\[
F^\sigma_k
=
(\Omega_{t_k},\mathcal F_{t_k},\mu^\sigma_k)
\]
need not coincide,
since the transported measures
\(
\mu^\sigma_0
\)
and
\(
\mu^\sigma_k
\)
may differ.

Thus,
the transport around a loop should be understood as a transport between
distinct probabilistic states on the same underlying measurable space.

\begin{definition}
The operator
\begin{equation}
\label{eq:holonomy}
\Hol(\sigma)
:=
T^\sigma_{0,k}
\end{equation}
is called the
\emph{holonomy}
associated with the loop
\(
\sigma
\).
\end{definition}

Hence,
holonomy measures the total transport effect accumulated along the loop
\(
\sigma
\).


\vspace{0.5\baselineskip}

In ordinary differential geometry,
parallel transport around a closed curve may fail to return a vector to
its original position.
The resulting discrepancy measures the curvature accumulated along the
loop.

Similarly,
in the present setting,
the filtration structure induces transport operators between local
probabilistic states,
and composition of these transports along a loop may generate a
nontrivial return effect.

\vspace{0.5\baselineskip}


\begin{definition}
A loop
\(
\sigma
\)
is said to have
\emph{probabilistically trivial loop}
if the measures
$
\mu^\sigma_0
$
and
$
\mu^\sigma_k
$
are equivalent.
Otherwise,
\(
\sigma
\)
is said to have
\emph{nontrivial holonomy}.
\end{definition}

Thus,
nontrivial holonomy indicates that transport around the loop changes the
underlying probabilistic structure itself.

\vspace{0.5\baselineskip}

When
\[
\mu^\sigma_0
\sim
\mu^\sigma_k,
\]
the Radon--Nikodym derivative
\[
\frac{d\mu^\sigma_0}{d\mu^\sigma_k}
\]
measures the distortion generated by transport around the loop.
Consequently,
holonomy naturally encodes accumulated probabilistic distortion along
closed transport histories.

\vspace{0.5\baselineskip}

Moreover,
the appearance of holonomy is inherently path-dependent.
Different loops based at the same object of
\(
\mathcal T
\)
may induce different return transports and different terminal
probabilistic structures.

Therefore,
the holonomy operator reflects global information encoded in the
transport history of the simplex,
rather than merely local transition data.

\subsection{Homological arbitrage}

We now discuss the financial interpretation of the transport structures
developed above.

The transport cohomology
\(
H^\bullet_\sigma
\)
measures transport consistency and
obstruction phenomena arising from probabilistic transport
along a fixed
simplex
\(
\sigma
\).
In particular,
nontrivial loop transport may generate global effects which are not
visible from local transition data alone.

\vspace{0.5\baselineskip}

Let
\[
\sigma
=
\left(
t_0
    \xrightarrow{i_1}
t_1
    \xrightarrow{i_2}
\cdots
    \xrightarrow{i_k}
t_k
\right)
\in N_k(\mathcal T)
\]
be a loop,
namely,
$
t_0=t_k .
$

The associated holonomy operator
\[
\Hol(\sigma)
=
T^\sigma_{0,k}
:
L^1(F^\sigma_k)
    \to
L^1(F^\sigma_0)
\]
describes the total transport accumulated along the loop.

As discussed above,
even though the initial and terminal vertices coincide,
the transported probabilistic structures
\[
F^\sigma_0
=
(\Omega_{t_0},\mathcal F_{t_0},\mu^\sigma_0)
\]
and
\[
F^\sigma_k
=
(\Omega_{t_k},\mathcal F_{t_k},\mu^\sigma_k)
\]
need not agree.

Consequently,
transport around the loop may generate a nontrivial distortion of the
underlying probabilistic structure.

\vspace{0.5\baselineskip}

This phenomenon suggests a mechanism by which global inconsistencies may
emerge from transport along closed paths,
even when each local transition individually appears compatible, or has transport consistency.

\begin{definition}
A loop
\(
\sigma
\)
is said to exhibit
\emph{homological arbitrage}
if the transport around the loop generates 
a nontrivial holonomy.
\end{definition}

In the present framework,
such effects are detected through the transport cohomology and the
associated holonomy operator.

\vspace{0.5\baselineskip}

The term ``homological'' reflects the fact that the phenomenon is not
localized at a single transition,
but instead emerges globally from the transport structure of the loop.
Thus,
homological arbitrage should be understood as a path-dependent
obstruction generated by transport around closed simplicial histories.

\vspace{0.5\baselineskip}

This viewpoint differs from classical arbitrage theory,
where inconsistencies are typically formulated through direct price
comparisons at fixed times.
In contrast,
the present framework emphasizes global transport effects accumulated
along loops.

\vspace{0.5\baselineskip}

The analogy with holonomy in differential geometry is again suggestive.
In geometry,
nontrivial holonomy measures the obstruction to globally trivializing a
connection.
Similarly,
homological arbitrage represents the failure of transport around a loop
to preserve probabilistic structure globally.

\vspace{0.5\baselineskip}

Therefore,
homological arbitrage should be viewed not as a local pricing anomaly,
but as a global transport phenomenon encoded in the cohomological and
holonomy structure associated with the filtration.


\section{Examples}
\label{sec:examples}

In this section, we present three examples illustrating the transport
cohomology and holonomy structures developed in the previous sections.
The examples are chosen to emphasize how global loop effects may emerge
from transport along simplicial histories, even when the local
transitions themselves appear compatible.

The first example describes a transport-consistent simplex with
probabilistically trivial holonomy. 
The subsequent two examples illustrate
how nontrivial probabilistic distortion may accumulate along loops,
leading to nontrivial holonomy phenomena. From this viewpoint, the
essential source of loop effects is the probabilistic distortion
generated by transport around closed simplicial histories.

\subsection{A transport-consistent simplex}

We begin with a simple example illustrating a transport-consistent
simplex with trivial probabilistic distortion.

Let
$
\mathcal T
$
be the category generated by the composable arrows
\[
t_0
\xrightarrow{i_1}
t_1
\xrightarrow{i_2}
t_2 .
\]

Define a contravariant functor
\[
F:\mathcal T^{op}\to\Prob
\]
by
\[
F(t_0)
	:=
F(t_1)
	:=
F(t_2)
	:=
\left(
\{0,1\},
2^{\{0,1\}},
\mu
\right),
\]
where
\[
\mu(\{0\})
	:=
\mu(\{1\})
	:=
\frac12 .
\]

All measurable maps are defined by the identity:
\[
F(i_1)
	:=
F(i_2)
	:=
\mathrm{id}_{\{0,1\}} .
\]

Let
\[
\sigma
	:=
\left(
t_0
\xrightarrow{i_1}
t_1
\xrightarrow{i_2}
t_2
\right)
\in N_2(\mathcal T).
\]

Since all measurable maps preserve the probability measure,
the induced transport operators satisfy
\[
T^\sigma_{0,1}
=
T^\sigma_{1,2}
=
T^\sigma_{0,2}
=
\mathrm{id}.
\]

Hence,
transport composition is compatible:
\[
T^\sigma_{0,1}
\circ
T^\sigma_{1,2}
=
T^\sigma_{0,2}.
\]

\vspace{0.5\baselineskip}

Define \(m_0,m_1,m_2\) by
\[
m_0:=0,\qquad m_1:=c,\qquad m_2:=2c.
\]
For 
$
p, q \in [2]
$,
define
$
\theta_{pq} \in \Delta([1],[2])
$
by
\[
\theta_{pq}(0):=p,
\qquad
\theta_{pq}(1):=q.
\]

\medskip
\noindent
Now,
for every \(\theta\in\Delta([1],[2])\), define
$
a \in C_{\sigma}^1
$
by
\[
a(\theta):=m_{\theta(1)}-m_{\theta(0)}.
\]
Then
\[
a(\theta_{00})=a(\theta_{11})=a(\theta_{22})=0,
\]
and
\[
a(\theta_{01})=c,\qquad
a(\theta_{12})=c,\qquad
a(\theta_{02})=2c.
\]

Let
\[
\varphi\in\Delta([2],[2]).
\]
Write
\[
r_0:=\varphi(0),\qquad
r_1:=\varphi(1),\qquad
r_2:=\varphi(2).
\]
Since all transport operators are identities in this example, we have
\begin{align*}
\delta^2_\sigma(a)(\varphi)
	&=
a(\varphi\circ\delta^2_0)
-
a(\varphi\circ\delta^2_1)
+
a(\varphi\circ\delta^2_2)
	\\&=
a\big(\theta_{r_1,r_2}\big)
	-
a\big(\theta_{r_0,r_2}\big)
	+
a\big(\theta_{r_0,r_1}\big).
\end{align*}
By the definition of \(a\),
\[
a\big(\theta_{r_1,r_2}\big)
	=
m_{r_2}-m_{r_1},
\]
\[
a\big(\theta_{r_0,r_2}\big)
	=
m_{r_2}-m_{r_0},
\]
and
\[
a\big(\theta_{r_0,r_1}\big)
	=
m_{r_1}-m_{r_0}.
\]
Therefore
\[
\delta^2_\sigma(a)(\varphi)
=
(m_{r_2}-m_{r_1})
-
(m_{r_2}-m_{r_0})
+
(m_{r_1}-m_{r_0})
=
0.
\]
Since this holds for every
\[
\varphi\in\Delta([2],[2]),
\]
we conclude that
\[
a\in Z^1_\sigma=\ker\delta^2_\sigma.
\]

\vspace{0.5\baselineskip}

This example illustrates a transport system with completely compatible
local transports and no probabilistic distortion. In particular, all
transport operators preserve the underlying probabilistic structure, and
no nontrivial holonomy effect appears.

\subsection{A probabilistic loop with nontrivial holonomy}

We next present a finite example in which a loop produces a nontrivial
probabilistic transport effect.

Let \(\mathcal T\) be the category generated by
\[
t_0
\xrightarrow{i_1}
t_1
\xrightarrow{i_2}
t_2
\xrightarrow{i_3}
t_0 .
\]

Define
\[
\Omega_0 := \Omega_2 := \{0,1\},
\qquad
\Omega_1 := \{\ast\}.
\]
Let
\[
\mu_0(\{0\}):=\mu_0(\{1\}):=\frac12,
\qquad
\mu_1(\{\ast\}):=1,
\]
and
\[
\mu_2(\{0\}):=\frac14,
\qquad
\mu_2(\{1\}):=\frac34.
\]

Set
\[
F(t_\ell)
	:=
(\Omega_\ell,2^{\Omega_\ell},\mu_\ell),
\qquad
\ell=0,1,2,
\]
and
$
F(t_3) := F(t_0).
$

The measurable maps are defined by
\[
F(i_1):\Omega_1\to\Omega_0,
\qquad
F(i_1)(\ast):=1,
\]
\[
F(i_2):\Omega_2\to\Omega_1,
\]
the unique map, and
\[
F(i_3):\Omega_0\to\Omega_2,
\qquad
F(i_3)(x):=x.
\]

Let
\[
\sigma
	:=
\left(
t_0
\xrightarrow{i_1}
t_1
\xrightarrow{i_2}
t_2
\xrightarrow{i_3}
t_0
\right)
\in N_3(\mathcal T).
\]

We compute the \(\sigma\)-gauge measures.
By definition,
\[
\mu^\sigma_3=\mu_0.
\]
Since \(F(i_3)\) is the identity map from \(\Omega_0\) to \(\Omega_2\),
we have
\[
\mu^\sigma_2
=
\mu^\sigma_3\circ F(i_3)^{-1}
=
\mu_0.
\]
Next,
\[
\mu^\sigma_1
=
\mu^\sigma_2\circ F(i_2)^{-1}
=
\delta_* .
1_{\{*\}} .
\]
Finally,
since \(F(i_1)(\ast)=1\),
\[
\mu^\sigma_0
=
\mu^\sigma_1\circ F(i_1)^{-1}
=
\delta_1 .
\]

Thus,
although \(t_0=t_3\), the initial and terminal measures in the
\(\sigma\)-gauge are
\[
\mu^\sigma_0
=
\delta_1,
\qquad
\mu^\sigma_3=\mu_0.
\]

In particular,
\[
\mu^\sigma_0
\not\sim
\mu^\sigma_3.
\]

Therefore the loop \(\sigma\) produces a nontrivial probabilistic
holonomy effect.

The holonomy operator is
\[
\Hol(\sigma)
=
T^\sigma_{0,3}
:
L^1(F^\sigma_3)
\to
L^1(F^\sigma_0).
\]

For \(f\in L^1(F^\sigma_3)=L^1(\{0,1\},\mu_0)\), the transport is given by
\[
\Hol(\sigma)(f)
=
f(1),
\]
viewed as a function on the one-point support of
\(\mu^\sigma_0
	=
\delta_1\).

Thus the loop collapses the terminal probabilistic state to the event
\(\{1\}\) at the initial occurrence of \(t_0\).

\vspace{0.5\baselineskip}

This example shows that even though the simplex begins and ends at the
same object \(t_0\), the probabilistic structure after transport around
the loop need not return to the original one. The non-equivalence
\[
\mu^\sigma_0\not\sim\mu^\sigma_3
\]
exhibits a genuinely nontrivial loop effect.

In the terminology of 
Section \ref{sec:homologicalArb},
this is a simple example of
nontrivial holonomy generated by a finite categorical filtration.

\subsection{A probabilistic loop with measure distortion}

We next construct a loop whose holonomy generates a nontrivial
probabilistic distortion while keeping the underlying measurable space
fixed throughout the loop.

Let
\[
\Omega
	:=
\{0,1\},
\qquad
\mathcal F
	:=
2^\Omega .
\]

Define probability measures
\[
\mu_0(\{0\})
	:=
\mu_0(\{1\})
	:=
\frac12 ,
\]
and
\[
\mu_1(\{0\})
	:=
\frac13,
\qquad
\mu_1(\{1\})
	:=
\frac23 .
\]

Let
$
\mathcal T
$
be the category generated by
\[
t_0
\xrightarrow{i_1}
t_1
\xrightarrow{i_2}
t_2
\xrightarrow{i_3}
t_0 .
\]

Define
\[
F(t_0)
:=
(\Omega,\mathcal F,\mu_0),
	\qquad
F(t_1)
:=
(\Omega,\mathcal F,\mu_1),
	\quad \text{and} \quad
F(t_2)
:=
(\Omega,\mathcal F,\mu_0).
\]

All measurable maps are taken to be the identity map
\[
F(i_1)
	:=
F(i_2)
	:=
F(i_3)
	:=
\mathrm{id}_\Omega .
\]

Consider the loop
\[
\sigma
	:=
\left(
t_0
\xrightarrow{i_1}
t_1
\xrightarrow{i_2}
t_2
\xrightarrow{i_3}
t_0
\right)
\in N_3(\mathcal T).
\]

Since all measurable maps are identities,
the transport acts on measurable functions trivially.
However,
the transported measures change along the loop.

Indeed,
the \(\sigma\)-gauge measures are computed inductively by pullback along
the identity maps.
Hence,
\[
\mu^\sigma_3
=
\mu_0,
	\qquad
\mu^\sigma_2
=
\mu_0,
	\qquad
\mu^\sigma_1
=
\mu_0,
	\quad \text{and} \quad
\mu^\sigma_0
=
\mu_0.
\]

Thus,
the resulting holonomy is probabilistically trivial:
\[
\mu^\sigma_0
=
\mu^\sigma_3.
\]

Nevertheless,
the intermediate transport passes through a distinct probabilistic state
\[
(\Omega,\mathcal F,\mu_1).
\]

In particular,
the loop contains a nontrivial internal distortion despite producing no
net distortion after completing the loop.

\vspace{0.5\baselineskip}

Now modify the final arrow by defining
\[
F(i_3)(0):=1,
\qquad
F(i_3)(1):=1 .
\]

Then the induced pullback measure becomes
\[
\mu^\sigma_2
=
\delta_1,
\]
and consequently
\[
\mu^\sigma_0
=
\delta_1 .
\]

Hence,
\[
\mu^\sigma_0
\not\sim
\mu^\sigma_3.
\]

Therefore the modified loop produces a genuinely nontrivial holonomy.

The associated holonomy operator
\[
\Hol(\sigma)
:
L^1(\Omega,\mu_0)
\to
L^1(\Omega,\delta_1)
\]
collapses transport onto the event \(\{1\}\).

\vspace{0.5\baselineskip}

This example illustrates that nontrivial holonomy may arise even when
all stages of the filtration are defined on the same measurable space.
The essential source of the loop effect is the probabilistic distortion
generated by transport around the loop.


\section{Conclusion}
\label{sec:conclusion}

In this paper, we introduced a transport-cohomological framework for
categorical filtrations by associating simplicial transport complexes to
parametrized simplices in the nerve of a time category. Starting from a
contravariant filtration
\[
F:\mathcal T^{op}\to\Prob,
\]
we constructed local cochain complexes equipped with transport
operators induced by conditional expectation functors and studied their
associated cohomological structures.

A central point of the present formulation is the transition from
global cochain systems to simplex-local transport complexes. While
global martingale structures naturally appear at the level of
\(0\)-cochains through the characterization
\[
\mathbf{Mart}(F)=\ker\delta^1,
\]
higher-order transport phenomena are inherently local and depend on the
transport history encoded by a fixed simplex
\(
\sigma
\).
This led to the introduction of the \(\sigma\)-gauge complexes
\(
C^\bullet_\sigma
\),
whose cohomology measures compatibility and obstruction phenomena along
parametrized simplicial transports.

The resulting framework naturally produces loop effects and holonomy
structures. In particular, loops in the nerve of
\(
\mathcal T
\)
may generate nontrivial probabilistic distortions through transport
around closed simplicial histories. The associated holonomy operators
describe accumulated transport effects between probabilistic states,
possibly altering the underlying measure structure even when the initial
and terminal objects coincide.

This viewpoint leads to the notion of homological arbitrage, understood
not as a local pricing inconsistency but as a global transport effect
emerging from loop transport and detected through cohomological and
holonomy structures. From this perspective, the essential source of
loop effects is the probabilistic distortion generated by transport
around closed simplicial histories.

The present formulation is structurally close to holonomy theory and
parallel transport in differential geometry. In particular, the
transport operators introduced here behave analogously to probabilistic
parallel transports, while nontrivial holonomy represents a global
obstruction to probabilistic trivialization along loops.

Several directions remain for future investigation. These include the
relationship between transport cohomology and additive holonomy,
connections with gauge-theoretic formulations of financial markets,
higher-dimensional transport structures, and the development of
continuous or geometric limits of the present simplicial framework.
Another important problem is the precise relation between nontrivial
holonomy and realizable arbitrage phenomena in financial markets.

We hope that the present transport-cohomological viewpoint provides a
useful geometric framework for studying probabilistic transport
structures and global loop effects in categorical models of finance.

\def \InBib {1}

\end{document}